\documentclass[a4paper]{article}
\usepackage{graphicx}
\usepackage{authblk}
\usepackage[authoryear]{natbib}
\usepackage{amsmath,amssymb,amsthm}
\newtheorem{theorem}{Theorem}[section]
\newtheorem{corollary}{Corollary}[theorem]
\newtheorem{lemma}[theorem]{Lemma}
\usepackage[hidelinks]{hyperref}
\begin{document}

\title{Deep Neural Network for Learning to Rank Query-Text Pairs}
    \author[1,2]{Baoyang \textsc{Song}\thanks{The work was conducted while the author was doing internship at the Data \& AI Lab of BNP Paribas.}}

    \affil[1]{Department of Computer Science, \'Ecole Polytechnique}
    \affil[2]{Data \& AI Lab, BNP Paribas}

\maketitle

\begin{abstract}
This paper considers the problem of document ranking in information retrieval systems by Learning to Rank.
We propose ConvRankNet combining a Siamese Convolutional Neural Network encoder and the RankNet ranking model 
which could be trained in an end-to-end fashion.
We prove a general result justifying the linear test-time complexity of pairwise Learning to Rank approach.
Experiments on the OHSUMED dataset show that ConvRankNet outperforms systematically existing feature-based models.

\end{abstract}

\section{Introduction}
\label{sec:introduction}

This paper considers the problem of document ranking in information retrieval systems by Learning to Rank.
Traditionally, people used to hand-tune ranking models such as TF-IDF or Okapi BM25 \citep{Manning:2008:IIR:1394399}
which is not only time-inefficient but also tedious.
Learning to Rank, on the other hand, aims to fit automatically the ranking model using machine learning techniques.
In recent years, Learning to Rank draws much attention and quickly becomes one of the most active research areas in information retrieval.
A number of supervised and semi-supervised ranking models has been proposed and extensively studied. 
We refer to \citep{liu2011learning} for a detailed exposition.

Though successful, these Learning to Rank models are mostly ``feature based''.
In other words, given a query-document pair $(q, d)$,
the inputs to ranking models are vectors of form $v = \Phi(q, d)$ 
where $\Phi$ is a \emph{feature extractor}. 
Some widely used features are such as TF-IDF similarity or PageRank score.
However, feature based models suffers from many problems in practice.
On the one hand, feature engineering is generally non-trivial and requires many trial-and-error before finding distinctive features;
on the other hand, the computation of $\Phi(q,d)$ could be very challenging in a real use case.
For instance, in the popular Elasticsearch \citep{Gormley:2015:EDG:2904394}, 
there is no direct way to calculate only IDF (though the product TF$\cdot$IDF is readily available).

In this paper, we show how raw query-document pairs could be directly used to fit an existing feature-based ranking model.
We develop ConvRankNet, a strong Learning to Rank framework composed of: (1) Siamese Convolutional Neural Network (CNN) encoder, 
a module designed to, 
given query $q$ and two documents $d_i, d_j$, extract automatically feature vectors $\Phi(q, d_i)$ and $\Phi(q, d_j)$ and 
(2) RankNet, a successful three-layer neural network-based pairwise ranking model.
We prove also a general result justifying the linear test-time complexity of pairwise Learning to Rank approaches.
Our experiments show that ConvRankNet improves significantly state-of-the-art feature based ranking models.

\section{Related Work}
\label{sec:related}

Our approach is based on the pairwise Learning to Rank approach \citep{liu2011learning}.

Pairwise approach is extensively studied under supervised setting. 
As its name suggests, it takes a pair of documents $d_i$ and $d_j$ as input and is trained to predict if $d_i$ is more relevant than $d_j$.
\cite{Joachims:2002:OSE:775047.775067} uses ``clickthough log'' to infer pairwise preference and 
trains a linear Support Vector Machines (SVM) on the difference of feature vectors $v(q,d_i)-v(q,d_j)$.
\cite{Burges:2005:LRU:1102351.1102363} introduce a three-layer Siamese neural network with a probabilistic cost function 
which can be efficiently optimized by gradient descent.
Instead of working with non-smooth cost function, 
\cite{NIPS2006_2971} propose LambdaRank which model directly the gradient of an implicit cost function.
\cite{burges2010ranknet} introduces LambdaMART which is the boosted tree version of LambdaRank. 
LambdaMART is generally considered as the state-of-the-art supervised ranking model.

Under semi-supervised setting, however, there is considerably fewer work. 
\cite{Szummer:2011:SLR:2063576.2063620} make use of unlabeled data by breaking the cost function $C$ into two parts: 
$C_{l}$ that depends only on labeled data and $C_{u}$ that depends only on unlabeled ones.
They report a statistically significant improvement over its supervised counterpart on some benchmark datasets.

Recent years, CNN \citep{Krizhevsky:2012:ICD:2999134.2999257} achieves impressive performance
on many domains, including Natural Language Processing (NLP) \citep{DBLP:journals/corr/Kim14f}. 
It is shown that CNN is able to efficiently learn to embed sentences into low-dimensional vector
space on preserving important syntactic and semantic aspects. 
Moreover, CNN is able to be trained in an end-to-end fashion, \emph{i.e.}\ little preprocessing and feature engineering are required.
Therefore, people attempt to adapt CNN to build an end-to-end Learning to Rank model.

\cite{Severyn:2015:LRS:2766462.2767738} combine a CNN with a pointwise\footnote{
Pointwise approach treats the feature vector of query-document pairs $\phi(q, d)$ independently. 
and considers a regression or multi-class problem on the relevance $r$.}
model to rank short query-text pairs
and report state-of-the-art result on several Text Retrieval Conference (TREC) tracks.
Though successful, their approach has several drawbacks: (1)
both query and document are limited to a single sentence;
(2) the underlying Learning to Rank approach is pointwise, 
which is rarely used in practice.
Moreover, it is difficult to take advantage of pairwise preference 
provided by clickthough log and thus not practical for a real use case;
(3) the add of ``additional features'' to the join layer is questionable. 
Indeed, the method is claimed to be end-to-end, 
but additional features could be so informative that the feature maps learned by CNN
do not play an import role.

We are thus motivated to generalize the idea in \cite{Severyn:2015:LRS:2766462.2767738} 
and to build a real end-to-end framework whose underlying ranker is a successful Learning to Rank model.

\section{ConvRankNet}
\label{sec:convranknet}

Throughout the rest of paper, we denote $\mathcal{Q}$ the set of queries and $\mathcal{D}$ the set of documents. 
Given $q\in\mathcal{D}$, note $\mathcal{D}_q\subset\mathcal{D}$ the set of documents which ``match''\footnote{For example, $\mathcal{D}_q$ could
be the set of documents sharing at least one token with $q$.}$q$.
For $d_i, d_j\in\mathcal{D}_q$, we write $d_i\succ d_j$ if $d_i$ is more relevance than $d_j$ ($d_i\prec d_j$ is defined similarly) and 
$d_i\sim d_j$ if there is a tie.
Note further $p:\mathcal{D}_q\times\mathcal{D}_q\to\{-1, 0, 1\}$ the pairwise preference such that 
\begin{equation}
    p(d_i, d_j) = \begin{cases}
                    -1, &\textrm{ if } d_i \prec d_j\\
                    0, &\textrm{ if } d_i \sim d_j \\
                    +1, &\textrm{ if } d_i \succ d_j
                \end{cases}
                \label{eqn:preference-function}
\end{equation}

In the following we describe our system ConvRankNet for ranking short query-text pairs in an end-to-end way 
which consists of 
(1) Siamese CNN Encoder, a module designed to extract automatically feature vectors from query and text and 
(2) RankNet, the underlying ranking model.
Figure \ref{fig:convranknet} gives an illustration of ConvRankNet.

\begin{figure}[htbp!]
    \includegraphics[width=\textwidth]{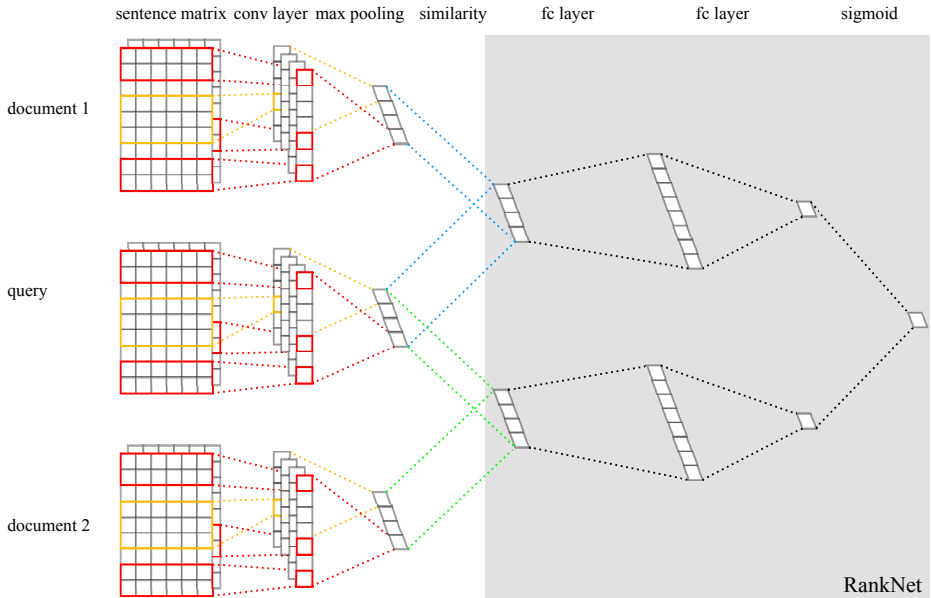}
    \caption{ConvRankNet structure.}
    \label{fig:convranknet}
\end{figure}

\subsection{Siamese CNN Encoder}

The Siamese CNN Encoder extracts feature vectors.
As shown in Figure \ref{fig:convranknet}, the encoder consists of three sub-networks sharing the same weights (a.k.a.\ Siamese network \citep{NIPS1993_769}).
It is made up of the following major components: 
sentence matrix, convolution feature maps, activation units, pooling layer and similarity measure.

\paragraph{Sentence Matrix}
Given a sentence $s = w_1\ldots w_N$,
the sentence matrix $S\in\mathbb{R}^{N \times d}$ is such that each row is the embedding of a word (or $n$-gram)
into a $d$-dimensional space by looking up a pre-trained word embedding model.

\paragraph{Convolution Feature Maps, Activation and Pooling}
Convolutional layer is used to extract discriminative patterns in the input sentence.
A $2d$-filter of size $m\times d$ is applied on a sliding window of $m$ rows of $S$ 
representing $m$ consecutive words (or $n$-grams). 
Note that the filter is of the same width $d$ as the sentence matrix, 
therefore, a column vector $v\in\mathbb{R}^{N+m-1}$ is produced.
Formally, the $i$-th component of $c$ is such that 
\begin{equation}
    v_i = f*S_{i:i+m-1} = \sum (f\odot S_{i:i+m-1}) + b
\end{equation}
where $b$ is a bias.
An non-linear activation unit is applied element-wise on $v$ which permits the network to learn non-linearity.
A number of activation units are widely used in many settings, 
in the scope of ConvRankNet, the rectified linear (ReLU) function is privileged.
The output of activation unit is further passed to a max-pooling layer. 
In other words, $v$ is represented by $\|\textrm{ReLU}(v)\|_{\infty}$. 
In practice, a set of filters of different size $\{f_1, \ldots, f_n\}$ are used to produce feature maps
$\{v_1, \ldots, v_n\}$. 
Each $v_i$ is passed individually through the activation unit and max pooling layer so that 
we have a vector
\begin{equation}
    \left\{\|\textrm{ReLU}(v_1)\|_{\infty}, \ldots,  \|\textrm{ReLU}(v_n)\|_{\infty}\right\}
\end{equation}
in the end.

\paragraph{Similarity Mesure}
Given $q\in\mathcal{Q}, d_i, d_j\in\mathcal{D}_q$, the encoder produces three vectors $v_q, v_{d_i}$ and $v_{d_j}$ respectively.
In order to feed RankNet, two feature vectors $\Phi(v_q, v_{d_i})$ and $\Phi(v_q, v_{d_j})$ need to be further computed.
\cite{Severyn:2015:LRS:2766462.2767738} introduce a similarity matrix $M$ and defines 
$\Phi(q, d) = v_q^TMv_d$. 
However, such a choice is difficult to be fitted in a modern deep learning framework.
In ConvRankNet, we choose instead
\begin{equation}
    \Phi(v_q, v_d) = (v_q - v_d)^2
\end{equation}
where the square is element-wise.

The output of Siamese CNN Encoder, $\Phi(v_q, v_{d_i})$ and $\Phi(v_q, v_{d_j})$, are then piped to a standard RankNet.
We privilege RankNet for its simple implementation and yet impressive performance on benchmark datasets.
Our idea, however, is applicable to any pairwise ranking model.

\subsection{RankNet}

Proposed in \citep{Burges:2005:LRU:1102351.1102363}, 
RankNet quickly becomes a popular ranking model and is deployed in commercial search engines such as Microsoft Bing.
It is well studied in the literature. For sake of completeness, however, we describe briefly here its structure. 

For $d_i, d_j\in\mathcal{D}_q$, suppose that there exists a deterministic target probability 
$\bar{P}(d_i, d_j) = \mathbb{P}(d_i\succ d_j):=\bar{P}_{ij}$
such that 
\begin{equation}
    \bar{P}_{ij} = \frac{1+p(d_i, d_j)}{2}\in\{0, 0.5, 1\}.
\end{equation}
The objective is to learn a posterior probability distribution $P$ that is ``close'' to $\bar{P}$. 
A natural measure of closeness between probability distribution is the binary cross-entropy 
\begin{equation}
    C(d_i, d_j) = -\bar{P}_{ij}\log P_{ij} - (1-\bar{P}_{ij})\log(1-P_{ij}) := C_{ij},
    \label{eqn:general-loss}
\end{equation}
which is minimized when $P = \bar{P}$.
The posterior $P$ is modeled by the Bradley-Terry model
\begin{align}
    P(d_i\succ d_j) &= \frac{1}{1+\exp(-s_{ij})}\\
    P(d_i\prec d_j) &= 1 - P(d_i\succ d_j) = \frac{1}{1+\exp(s_{ij})}.
\end{align}
where $s$ is a score function $s: \mathcal{D}\to\mathbb{R}$ and $s_i = s(d_i)$ and $s_{ij} = s(d_i) - s(d_j)$. 

Under this assumption, the loss function (\ref{eqn:general-loss}) further becomes
\begin{equation}
    C_{ij} = -\bar{P}_{ij}s_{ij} + \log(1+\exp(s_{ij}))
    \label{eqn:ranknet-loss}
\end{equation}
$C_{ij}$ being convex, it is straightforward to optimize it by gradient descent. 

Since $d_i\succ d_j$ iff.\ $d_j \prec d_i$, without loss of generality we suppose 
that for $(d_i, d_j)$ we always have $d_i \succeq d_j$.
Moreover, \cite{Burges:2005:LRU:1102351.1102363} show that training on ties 
makes little difference. 
Therefore, we could consider only document pairs $(d_i, d_j)$ such that $d_i \succ d_j$. 

\subsection{Time Complexity}

In this section we discuss the time complexity of general pairwise Learning to Rank models (and in particular, ConvRankNet).

In pairwise approach, 
we generally consider a bivariate function $h: \mathcal{D}_q\times \mathcal{D}_q \to \mathbb{R}, 
(d_i, d_j) \mapsto h(d_i, d_j)$ such that $d_i \succ d_j$ iff.\ $h(d_i, d_j) > 0$.
$h$ is then used to construct the cost function.

It is clear that the training time complexity is $\mathcal{O}(|\mathcal{D}|^2)$
since every pair $(d_i, d_j)$ 
such that $p(d_i, d_j)\neq 0$ has to be considered. 
One may infer that the test cost is also quadratic 
since we have to evaluate $h$ on the collection of all possible pairs on test data 
and construct a consistent total order (For example, $h(d_1, d_2) > 0, h(d_2, d_3) > 0$ induces the total order $d_1 \succ d_2 \succ d_3$).

However, we argue that under a very loose assumption, 
a linear time\footnote{Here we do not take into account the sort of score, which is in general $\mathcal{O}(n\log n)$. 
However, if the score is of fix precision, on could use \emph{e.g.}\ Radix sort to achieve linear time complexity.} 
$\mathcal{O}(|\mathcal{D}|)$ actually suffices for constructing the total order. 
First recall the following result in graph theory:
\begin{lemma}
    The \emph{topological sort} $r$ of a directed graph $\mathcal{G}$ is \emph{an} order of 
    vertices such that all edges of $\mathcal{G}$ go from left to right in the order.
    It is shown that
    \begin{itemize}
        \item $\mathcal{G}$ is a directed acyclic graph (DAG) iff.\
    there exists a topological sort on $\mathcal{G}$ \citep{Skiena:2008:ADM:1410219}.
        \item the topological sort is unique iff.\ 
    there is a directed edge between each pair of consecutive vertices 
    in the topological order (\emph{i.e.}\ $\mathcal{G}$ has a Hamiltonian path) 
    \citep{Sedgewick:2011:ALG:2011916}.
\end{itemize}
    \label{lem:topological-sort}
\end{lemma}
We have then the following result:
\begin{theorem}
    Suppose that the hypothesis $h$ has no tie, \emph{i.e.}\ $h(d_i, d_j)\in\mathbb{R}_+^*$.
    If there exists $\psi: \mathbb{R}\to\mathbb{R}, f: \mathbb{R}\to\mathbb{R}$ such that 
    $h(d_i, d_j) = \psi\circ (f(d_i) - f(d_j))$ and $h(d_i, d_j)>0$ iff.\ $f(d_i) > f(d_j)$, 
    then the total order defined by $f(\cdot)$ on $\mathcal{D}_q$ is the same as that of $h(\cdot,\cdot)$ 
    on $\mathcal{D}_q\times \mathcal{D}_q$.
    \label{thm:pairwise-linear-time-testing}
\end{theorem}
\begin{proof}
    Consider the graph $\mathcal{G}=(\mathcal{V}, \mathcal{E})$ 
    induced by $h(\cdot, \cdot)$, \emph{i.e.}\ $(i, j)\in\mathcal{E}$ iff.\  $h(d_i, d_j) > 0$.
    Remark first that $\mathcal{G}$ is a DAG. If not, there exists $d_{i_1}\ldots d_{i_n}\in\mathcal{D}$ such that 
    $d_{i_1}\succ d_{i_2}\succ\ldots\succ d_{i_n}\succ d_{i_1}$. 
    Then $f(d_{i_1})> f(d_{i_2})>\ldots> f(d_{i_n})>f(d_{i_1})$, a contradiction. 
    By Lemma \ref{lem:topological-sort}, there exists a topological sort. 
    Without loss of generality note the sort $d_q = \{d_1, \ldots, d_n\}$.
    Since $h$ has no tie, $d_q$ is an Hamiltonian path, thus the topological sort is furthermore unique.
    It is easy to see that $r$ it is nothing but the sort with respect to $f(\cdot)$.
\end{proof}
\begin{corollary}
    ConvRankNet has linear test time.
\end{corollary}

\section{Experimental Results}
\label{sec:experiments}

In this section we evaluate ConvRankNet on standard benchmark datasets and compare it with standard RankNet and LambdaRank.

\subsection{Datasets}

Since ConvRankNet is an end-to-end model, 
we need datasets to which we have access to raw query and documents.

To the best of our knowledge, OHSUMED dataset\footnote{\url{http://mlr.cs.umass.edu/ml/machine-learning-databases/ohsumed/}} is the only freely available dataset. 
Subset of MEDLINE (a database on medical publications),
it consists of $106$ queries on $348566$ medical documents during 1987-1991. 
The relevance of $16140$ query-document pairs are provided by human assessors on three levels: 
non-relevant (n), partially relevant (p) and definitely relevant (d). 
In particular, non-relevant pairs are explicitly provided. 
For historical reasons, each query-document pair is judged independently by up to $3$ assessors. 
To avoid ambiguity, the highest relevance label is taken in our experiments.
\begin{table}[htbp!]
    \centering
    \caption{Partition of OHSUMED dataset.}
    \label{tab:partition-ohsumed}
    \begin{tabular}{cccccc}
    \hline\noalign{\smallskip}
    &$S_1$&$S_2$&$S_3$&$S_4$&$S_5$\\
    \noalign{\smallskip}\hline\noalign{\smallskip}
    Query id & 1-21 & 22-42 & 43-63 & 64-84 & 85-106\\
    \noalign{\smallskip}\hline
    \end{tabular}
\end{table}

To compare with classical feature-base models, 
we also use a synthesized version (where only feature vectors $\Phi(q,d)$ are accessible) of OHSUMED which is included in Microsoft's LETOR 3.0 dataset
\footnote{\url{https://www.microsoft.com/en-us/research/project/letor-learning-rank-information-retrieval/}.}.
As in LETOR 3.0, we partition raw OHSUMED dataset into $5$ folds as shown in Table \ref{tab:partition-ohsumed}.

\subsection{Experimental Setup}

All models are implemented in PyTorch framework.

In general, query and documents are not of the same length.
Though PyTorch uses dynamic graph and is capable of handling texts of various lengths, 
one could only train the network one query-document-document triple a time.
In order to perform batch training, both query and document are truncated to $100$ words with zero-padding.

We use ConceptNet Numberbatch\footnote{\url{https://github.com/commonsense/conceptnet-numberbatch}}\cite{DBLP:journals/corr/SpeerCH16} 
as the default word embedding.
Part of the ConceptNet open data project, 
ConceptNet Numberbatch consists of state-of-the-art semantic vectors that can be used directly as representation of word meanings.
It is built using an ensemble that combines data from ConceptNet, 
word2vec, GloVe, and OpenSubtitles 2016, using a variation on retrofitting.
Several benchmarks show that ConceptNet Numberbatch outperfoms word2vec 
and GloVe.
ConceptNet Numberbatch includes not only vectorization of single word but also that of some bigrams and trigrams.
Bigrams and trigrams are semantically more informative. 
For example, ``la carte'' is clearly better characterized than ``la'' and ``carte'' separately.
To exploit $n$-grams, we use a greedy approach in mapping text to vector matrix, 
\emph{i.e.}\ we extend as long as possible the $n$-gram to map to vector. 
For example, the following toy ``sentence'' \texttt{hello world peace}
would be segmented as $[\texttt{hello world}, \texttt{peace}]$, 
even though $[\texttt{hello}, \texttt{world peace}]$ is also possible.
One possible extension to our work is then to find the semantically optimal segmentation of sentence.
In our network, unknown words are mapped to a random vector and zero padding is mapped to zero vector.

Three different filters of size $3\times 300, 4\times 300, 5\times 300$, with $10$ copies each, are used for the convolution layer
so that the input for RankNet is a $30$-dimensional vector. 
In order to prevent overfitting, during training stage a drop-out layer\cite{Srivastava:2014:DSW:2627435.2670313} 
with $p=0.5$ is used after the max-pooling layer.

RankNet, LambdaRank and ConvRankNet are all trained for $500$ epochs with learning rate $0.00001, 0.001, 0.001$ respectively.
All tests were performed on a Ubuntu 16.04.4 server with  
$2\times$ Xeon $3.2$GHz CPU, $256$GB RAM and Tesla P100 $16$GB GPU.
Cross-validations are run in parallel with the help of GNU Parallel \citep{Tange2011a}.

Normalized Discounted Cumulative Gain at truncation level $k$ (NDCG@$k$) is used as the evaluation measure. 
NDCG@$k$ is a multi-level ranking quality measure widely used in previous work. 
It ranges from $0$ to $1$ with $1$ for the \emph{perfect ranking}. We refer to \citep{Manning:2008:IIR:1394399} for a detailed presentation.

\subsection{Results}
\begin{table}[htbp!]
    \centering
    \caption{$5$-fold cross-validation of NDCG@$k$ on test set. }
    \label{tab:ranknet-convranknet-benchmark}
    \begin{tabular}{llllll}
        \hline\noalign{\smallskip}
        &  NDCG@1 &  NDCG@2 &  NDCG@3 &  NDCG@4 &  NDCG@5 \\
        method &         &         &         &         &         \\
        \noalign{\smallskip}\hline\noalign{\smallskip}
        ConvRankNet &  0.5479 &  0.5265 &  \textbf{0.5204} &  \textbf{0.5241} &  \textbf{0.5204} \\
        LambdaRank &  0.5677 &  \textbf{0.5267} &  0.4942 &  0.4884 &  0.4780 \\
        RankNet &  \textbf{0.5737} &  0.5362 &  0.5128 &  0.4898 &  0.4746 \\
        \noalign{\smallskip}\hline
    \end{tabular}
    \begin{tabular}{llllll}
        \hline\noalign{\smallskip}
        &  NDCG@6 &  NDCG@7 &  NDCG@8 &  NDCG@9 &  NDCG@10 \\
        method &         &         &         &         &          \\
        \noalign{\smallskip}\hline\noalign{\smallskip}
        ConvRankNet &  \textbf{0.5179} &  \textbf{0.5109} &  \textbf{0.5139} &  \textbf{0.5122} &   \textbf{0.5132} \\
        LambdaRank &  0.4681 &  0.4604 &  0.4552 &  0.4553 &   0.4503 \\
        RankNet &  0.4648 &  0.4608 &  0.4560 &  0.4493 &   0.4461 \\
        \noalign{\smallskip}\hline
    \end{tabular}
\end{table}

$5$-fold cross-validation is performed on both datasets. 
Table \ref{tab:ranknet-convranknet-benchmark}
reports NDCG@$k$ at all truncation levels. 
It is clear that ConvRankNet outperforms RankNet and LambdaRank especially for large truncation level $k$.

\begin{table}[htbp!]
    \centering
    \caption{$p$-value of two-tailed Wilcoxon signed-rank test.}
    \label{tab:ohsumed-p-value}
    \begin{tabular}{llll}
        \hline\noalign{\smallskip}
        & RankNet & LambdaRank & ConvRankNet\\
        \noalign{\smallskip}\hline\noalign{\smallskip}
        RankNet & - & 0.575 & 0.021\\
        LambdaRank & - &- & 0.012\\
        ConvRankNet & - & - &-\\
        \noalign{\smallskip}\hline
    \end{tabular}
\end{table}
A two tailed Wilcoxon signed-rank test \citep{citeulike:14057436} is performed on these values. 
As we can see from Table \ref{tab:ohsumed-p-value}, the improvement of ConvRankNet over RankNet and LambdaRank 
is statistically significant at $5\%$ level.
Moreover, according to \citep{liu2011learning}, ConvRankNet also outperfoms systematically existing methods.

\section{Discussion and Conclusions}
\label{sec:discussion}

In this paper, we proposed ConvRankNet, an end-to-end Learning to Rank model which directly takes raw query and documents as input. 
ConvRankNet is shown to have linear test time and thus applicable in real-time use cases.
Our results indicate that it outperforms significant existing methods on OHSUMED dataset.

Future work could aim to study the generalization of the underlying RankNet module to other stronger neural network based model (such as LambdaRank).

\section*{Acknowledgements}
    The author would like to thank Dr.\ Michalis \textsc{Vazirigiannis} of \'Ecole Polytechnique for his valuable suggestions.

    The author would also like to thank Mr.\ Geoffrey \textsc{Scoutheeten} and Mr.\ \'Edouard d'\textsc{Archimbaud} of BNP Paribas for their valuable comments.

    The work is supported by the Data \& AI Lab of BNP Paribas.

\bibliographystyle{unsrtnat}
\bibliography{reference}   
\end{document}